\documentclass{article}
\usepackage{ijcai25}

\usepackage{times}
\usepackage{soul}
\usepackage{url}
\usepackage[hidelinks]{hyperref}
\usepackage[utf8]{inputenc}
\usepackage[small]{caption}
\usepackage{graphicx}
\usepackage{amsmath}
\usepackage{amsthm}
\usepackage{booktabs}
\usepackage{algorithm}
\usepackage{algorithmic}
\usepackage[switch]{lineno}
\usepackage{amssymb,mathabx}
\usepackage{xcolor}

\usepackage{fancyhdr}
\fancypagestyle{firstpage}{
  \fancyhf{}  
  \fancyhead[R]{\textit{Accepted to IJCAI 2025.}}
  \fancyfoot[C]{\thepage} 
  
}
\newtheorem{theorem}{Theorem}
\newtheorem{definition}{Definition}
\newtheorem{assumption}{Assumption}
\newtheorem{lemma}{Lemma}

\newcommand{\norm}[1]{||#1||}

\title{KP-PINNs: Kernel Packet Accelerated Physics Informed Neural Networks}

\author{
Siyuan Yang $^{1,2}$
\and
Cheng Song $^1$\and
Zhilu Lai $^{2,3}$\And
Wenjia Wang $^1$  \thanks{Corresponding author.}  \\
\affiliations
$^1$ Data Science and Analytics, The Hong Kong University of Science and Technology (Guangzhou) \\
$^2$ Internet of Things, The Hong Kong University of Science and Technology (Guangzhou) \\
$^3$ Department of CEE, The Hong Kong University of Science and Technology  \\
\emails
\{syang724, csong750\}@connect.hkust-gz.edu.cn,
zhilulai@ust.hk,
wenjiawang@hkust-gz.edu.cn
}

\begin{document}

\thispagestyle{firstpage}

\maketitle

\begin{abstract}
Differential equations are involved in modeling many engineering problems. Many efforts have been devoted to solving differential equations. Due to the flexibility of neural networks, Physics Informed Neural Networks (PINNs) have recently been proposed to solve complex differential equations and have demonstrated superior performance in many applications. While the $L_2$ loss function is usually a default choice in PINNs, it has been shown that the corresponding numerical solution is incorrect and unstable for some complex equations. In this work, we propose a new PINNs framework named Kernel Packet accelerated PINNs (KP-PINNs), which gives a new expression of the loss function using the Reproducing Kernel Hilbert Space (RKHS) norm and uses the Kernel Packet (KP) method to accelerate the computation. Theoretical results show that KP-PINNs can be stable across various differential equations. Numerical experiments illustrate that KP-PINNs can solve differential equations effectively and efficiently. This framework provides a promising direction for improving the stability and accuracy of PINNs-based solvers in scientific computing.   
\end{abstract}

\section{Introduction} \label{s:intro}
Differential equations, including ordinary differential equations (ODEs) and partial differential equations (PDEs), are widely applied in engineering for modelling, analyzing, and optimizing dynamic systems \cite{kreyszig2007advanced}. They provide a mathematical framework for describing natural phenomena and solving complex problems. 
Differential equations are widely used in simulating fluid dynamics \cite{ragheb1976computational}, electromagnetics \cite{deschamps1981electromagnetics}, heat transfer \cite{isachenko1980heat} and acoustics \cite{kaltenbacher2018computational}
Differential equations solutions are crucial for accurately predicting and optimizing the behaviour of dynamic systems, ensuring effective design and operation in engineering applications. 

Because analytical solutions of differential equations are often difficult to obtain in practice, researchers seek numerical methods instead. For example, Euler's method and Runge-Kutta method \cite{atkinson2009numerical} can be applied to solving ODEs, and finite element method (FEM), finite difference method (FDM), finite volume method (FVM), and spectral methods can be used to solve PDEs. These methods provide a foundation for solving various differential equations encountered in engineering and scientific applications. However, such methods require substantial computational resources and memory, particularly in high-dimensional or complex physical scenarios. Moreover, they lack flexibility, as adapting to different physical problems or parameter changes often needs significant model reconfiguration \cite{karniadakis2021physics}.

With the development of deep learning, Physics informed neural networks (PINNs) incorporate physical laws into training neural networks to solve differential equations \cite{kianiharchegani2023data}, particularly PDEs, by effectively combining deep learning with domain-specific knowledge. PINNs have strong flexibility and can be effectively combined with existing traditional methods. However, PINNs can lead to incorrect solutions for some complex PDEs, affecting the accuracy of the solutions obtained \cite{zhang2024physics}.
Several methods are proposed to address this issue. For example, \cite{haitsiukevich2022improved} used ensemble agreement for domain expansion, and the improved algorithm can make the training of PINNs more stable. \cite{hu2021extended} theoretically analyzed the convergence and generalization of extended PINNs. \cite{yu2022gradient} proposed gradient-enhanced PINNs, which utilize the gradient information of PDEs residuals and embed the gradient into the loss function to improve the effectiveness of solving PDEs. Other works include \cite{jung2023exploring,forootani2024gs,aliakbari2023physics,yuan2022pinn,lin2022multi}.

However, it has been noticed that the original PINNs are not stable \cite{wang20222}, due to the default $L_2$ loss function. To address the instability of the original PINNs, \cite{wang20222} proposed the $L_\infty$ loss function and proposed an adversarial training method. However, the computation of $L_\infty$ norm is complicated, and the adversarial training is time-consuming. Another class of loss functions is based on the Sobolev norm, which is proposed by \cite{son2021sobolev}, and the corresponding method is named as Sobolev-PINNs. However, the computation of the Sobolev norm in Sobolev-PINNs is complex, and accurate derivative information is hard to get. 

In this work, we propose a novel form of loss function using the Reproducing Kernel Hilbert Space (RKHS) norm. Under specific choice of the kernel function, the RKHS norm is equivalent to the (tensored) Sobolev norm, thus addressing the stability issue of the original PINNs. The proposed method can also avoid directly computing the derivative because the RKHS norm of a function can be approximated by the RKHS norm of its interpolation, while the later has a close form. To compute the RKHS norm, we employ Kernel Packet (KP) method \cite{chen2022kernel}, which can efficiently compute the inverse of the kernel matrix. Therefore, KP accelerated PINNs (KP-PINNs) can solve differential equations effectively and efficiently. We also provide theoretical guarantee for our proposed method.

The rest is organized as follows.
Section \ref{s:prel} briefly introduces differential equations, PINNs and KP.
Section \ref{s:method} analyses the proposed KP-PINNs algorithm.
Section \ref{s:theor} gives some theoretical results.
Section \ref{s:exper} shows four numerical examples to demonstrate the performance of the KP-PINNs algorithm.
Section \ref{s:conclusion} gives this research's discussion and future work.

\section{Preliminaries} \label{s:prel}

\subsection{Differential Equations}  \label{s:prel.1}
Differential equations are mathematical equations that describe the relationship between a function and its derivatives, representing how a quantity changes over time or space. They are broadly classified into ordinary differential equations (ODEs), which involve functions of a single variable, and partial differential equations (PDEs), which involve functions of multiple variables. Differential equations are foundational in modeling dynamic systems and natural phenomena. By capturing the principles of change and interaction, they provide powerful tools for analyzing and solving complex problems.
In general, a differential equation system can be expressed \cite{evans2022partial} by:
\begin{equation} \label{eq: normalized pde}
\left\{\begin{array}{ll}
\mathcal{L}_{\theta} u(\mathbf{x}) = f_{\theta}(\mathbf{x}), & \mathbf{x} \in \Omega \subset \mathbb{R}^{d},  \\ 
\mathcal{B}_{\theta} u(\mathbf{x}) = g_{\theta}(\mathbf{x}), & \mathbf{x} \in \partial \Omega,
\end{array}\right. 
\end{equation}
where $\Omega$ is an open domain with boundary $\partial \Omega$, $\mathcal{L}_{\theta}$ represents a partial differential operator, $\mathcal{B}_{\theta}$ denotes the boundary condition, $f_{\theta}$ and $g_{\theta}$ are two known functions defined on $\Omega$ and $\partial\Omega$, respectively, $\theta$ is the parameter, and $u$ is an unknown function. The goal of solving a differential equation system is to find $u$ such that \eqref{eq: normalized pde} holds. Usually, \eqref{eq: normalized pde} is called ODE if $d=1$, and PDE if $d>1$. An example of PDEs is the second-order linear PDE, with
\begin{equation}\label{eq: second-order linear PDEs}
    \mathcal{L}_{\theta} u = \sum_{i=1}^d \sum_{j=1}^d \alpha_{i,j} \frac{\partial^2u}{\partial x_i\partial x_j} + \sum_{k=1}^d \beta_k \frac{\partial u}{\partial x_k} + \gamma u,
\end{equation}
$\mathbf{x} = (x_1,...,x_d)^\top \in \mathbb{R}^{d}$, and the parameters are $\alpha_{i,j}, \beta_k, \gamma$ for $i,j,k\in \{1,\ldots,d\}$. 

Besides the linear PDEs, where $\mathcal{L}_{\theta}$ is a linear differential operator, nonlinear PDEs, in which the solution or its derivatives appear nonlinearly, make them more complex to solve compared to linear PDEs. Because of the flexibility of nonlinear PDEs, it has been widely applied in fluid dynamics \cite{yan2025deep}, optics \cite{evans1989pde}, and biology \cite{ghergu2011nonlinear}. One particular example is the Hamilton-Jacobi-Bellman (HJB) equations, which are widely applied in electro-hydraulic systems \cite{guo2022optimal}, energy systems \cite{sieniutycz2000hamilton} and finance \cite{witte2011penalty}. 
The HJB equations characterise the optimal control strategy for a dynamic system under certain constraints and objectives. Specifically, the form of the HJB equations \cite{yong2012stochastic} can be expressed by
\begin{equation} \label{eq: HJB}
\begin{cases}
\partial_t u(\mathbf{x}, t) + \frac12 \sigma^2 \Delta u(\mathbf{x}, t) + \min_{k \in \mathcal{K}}, \\ 
\qquad [r(\mathbf{x}, k(\mathbf{x}, t)) + \nabla u \cdot k_t] = 0, \\
u(\mathbf{x}, T) = g(\mathbf{x}), (\mathbf{x}, t) \in \mathbb{R}^d \times [0,T],
\end{cases}
\end{equation}
where $u(\mathbf{x}, t)$ is the value function, $k(\mathbf{x}, t)$ is the given control function, $r(\mathbf{x}, k)$ is the cost rate during the process and $g(\mathbf{x})$ is the final cost at the terminal state \cite{wang20222}. If the cost rate function is $r(\mathbf{x}, k) = a_1 |k_1|^{b_1} + \cdots + a_d |k_d|^{b_d} - h(\mathbf{x}, t)$ then the HJB equations can be expressed \cite{yong2012stochastic} as
\begin{equation} \label{eq: reformulated HJB}
\begin{cases}
\mathcal{L}_\text{HJB}u := \partial_t u(\mathbf{x}, t) + \frac12 \sigma^2 \Delta u(\mathbf{x}, t) - \\ 
\qquad \sum_{i=1}^d A_i |\partial_{\mathbf{x}_i}u|^{c_i} = h(\mathbf{x}, t),  \\
\mathcal{B}_\text{HJB}u := u(\mathbf{x}, T) = g(\mathbf{x}), (\mathbf{x}, t) \in \mathbb{R}^d \times [0,T],
\end{cases}
\end{equation}
where $A_i = (a_i b_i)^{-\frac{1}{b_i - 1}} - a_i (a_i b_i)^{-\frac{b_i}{b_i-1}} > 0$ with $c_i = \frac{b_i}{b_i-1} > 1$, $\mathcal{L}_0 u:= \frac{\partial}{\partial t} u - \Delta u$ is a linear operator, and $\tilde{\mathcal{L}}_\text{HJB}u:= \mathcal{L}_0 u + \sum_{i=1}^d A_i|\partial_iu|^{c_i}$ is a nonlinear operator.

In the forward problem of differential equations, the goal is to obtain a function $u$ such that \eqref{eq: normalized pde} is satisfied, where all the operators and functions $\mathcal{L}_\theta$, $\mathcal{B}_\theta$, $h_\theta$ and $g_\theta$ are known. In practice, there are many scenarios where the parameter $\theta$ is unknown and needs to be estimated based on the equation form and observed data. This is called inverse problems, where the goal is to recover the parameters. In this work, we consider both forward problems and inverse problems.

\subsection{Physics Informed Neural Networks} \label{s:prel.2}
Physics informed neural networks (PINNs) provide a flexible way for numerically solving the differential equation systems \eqref{eq: normalized pde}. 

For notational convenience, let
\begin{equation} \label{eq: h12}
h_1 = \mathcal{L} u - f_{\theta}, h_2 = \mathcal{B} u - g_{\theta}.
\end{equation}
Clearly, $u$ is the solution to \eqref{eq: normalized pde} if and only if $h_1=h_2=0$. Therefore, for the forward problem, it is natural to minimize the loss function
\begin{align}\label{eq_L2loss}
    {\rm Loss}_{L_2} = \norm{h_1}_{L_2(\Omega)}^2 + \norm{h_2}_{L_2(\partial \Omega)}^2,
\end{align}
which was considered by \cite{raissi2019physics}. Choose $\mathbf{x}^i\in \Omega$, $i=1,...,N_{\mathcal{L}}$ and $\mathbf{s}^i\in \partial \Omega$, $i=1,...,N_{\mathcal{B}}$. Then \eqref{eq_L2loss} can be approximated by
\begin{align}\label{eq: MSE forward}
    {\rm Loss}(u) = {\rm Loss}_{\mathcal{L}}(u) + {\rm Loss}_{\mathcal{B}}(u),
\end{align}
with
\begin{align*}
    {\rm Loss}_{\mathcal{L}} = \frac{1}{N_{\mathcal{L}}} \sum_{i=1}^{N_{\mathcal{L}}} {h_1(\mathbf{x}^i)}^2 \quad {\rm and} \quad {\rm Loss}_{\mathcal{B}} = \frac{1}{N_{\mathcal{B}}} \sum_{i=1}^{N_{\mathcal{B}}} {h_2(\mathbf{s}^i)}^2.
\end{align*}
PINNs use a deep neural network, denoted by $u_{nn}$, to minimize the loss function in \eqref{eq: MSE forward}. 

For the inverse problems, suppose that we have observed function values at points $\mathbf{x}^i\in \Omega$, $i=1,...,N_{\mathcal{L}}$ and $\mathbf{s}^i\in \partial \Omega$, $i=1,...,N_{\mathcal{B}}$, denoted by $u_d(\mathbf{x}^i)$ and $u_d(\mathbf{s}^i)$, respectively. We define the loss function in the inverse problem as
\begin{equation}\label{eq: MSE inverse}
    {\rm Loss}_{{\rm Inv}}(u) =  {\rm Loss}(u) + {\rm Loss}_{\mathcal{D}}(u),
\end{equation}
where ${\rm Loss}(u)$ is as in \eqref{eq: MSE forward}, and
\begin{align*}
    {\rm Loss}_{\mathcal{D}}(u) =  \frac{1}{N_{\mathcal{L}}} \sum_{i=1}^{N_{\mathcal{L}}} {h_3(\mathbf{x}^i)}^2 + \frac{1}{N_{\mathcal{B}}} \sum_{i=1}^{N_{\mathcal{B}}} {h_3(\mathbf{s}^i)}^2
\end{align*}
with $h_3 = u - u_d$. The goal of solving an inverse problem is to solve
\begin{equation}
    \min_{\substack{\theta \in \Theta}} {\rm Loss}_{{\rm Inv}}(u),
\end{equation}
where $\Theta$ is the parameter space. In practice, the minimizers of the loss functions in both forward and inverse problems can be found via optimization algorithms such as Adam \cite{kingma2014adam} and L-BFGS \cite{byrd1995limited}.

\subsection{Kernel Packet} \label{s:prel.4}
Our proposed method relies on \emph{Kernel Packet} (KP), which was proposed in \cite{chen2022kernel}. 
In functional analysis, a Hilbert space of functions is termed a Reproducing Kernel Hilbert Space (RKHS) if the evaluation is a continuous linear operator at any point within the space.
Assume that $\Omega \subset \mathbb{R}^d$ is compact with Lipschtiz boundary, and $\Phi: \Omega \times \Omega \to \mathbb{R}$ is a symmetric positive definite kernel function \cite{wang2022gaussian,wang2020prediction}. Define the linear space
\begin{equation*}
F_\Phi(\Omega) = \left\{ \sum_{i=1}^N \beta_i \Phi(\cdot, x_i) : \beta_i \in \mathbb{R}, x_j \in \Omega, N \in \mathbb{N} \right\}
\end{equation*}
and equip this space with the bilinear form
\begin{equation*}
\left\langle \sum_{i=1}^N \beta_i \Phi(\cdot, x_i), \sum_{j=1}^M \gamma_j \Phi(\cdot, x_j') \right\rangle_K:= 
\sum_{i=1}^N \sum_{j=1}^M \beta_i \gamma_j \Phi(x_i, x_j'),
\end{equation*}
where the RKHS $\mathcal{N}_\Phi(\Omega)$ generated by kernel $\Phi$ is defined as the closure of $F_\Phi(\Omega)$ under the inner product $\langle \cdot, \cdot \rangle_\Phi$. The norm of $\mathcal{N}_\Phi(\Omega)$ is $\norm{f}_{\mathcal{N}_\Phi(\Omega)} = \sqrt{\langle f, f \rangle_{\mathcal{N}_\Phi(\Omega)}}$, where $\langle \cdot, \cdot \rangle_{\mathcal{N}_\Phi(\Omega)}$ is induced by $\langle \cdot, \cdot \rangle_\Phi$. 

Let $\mathbf{x}_1,...,\mathbf{x}_n\in \Omega\cup \partial \Omega$ be the points of interest. KP provides an efficient way to compute the inverse of the kernel matrix $\mathbf{K}:= (\Phi(\mathbf{x}_j,\mathbf{x}_k))_{jk} \in \mathbb{R}^{n\times n}$, which is essential in the application of the RKHSs, for example, prediction in kernel ridge regression, and computing the posterior variance in Bayesian optimization. 
In KP, it is assumed that the kernel function $\Phi$ is a Mat\'ern kernel function \cite{abramowitz1968handbook} as 
\begin{align}    \label{eq: K_{p+1/2}}
\Phi_{\nu}(\mathbf{x}_j,\mathbf{x}_k)=\exp\left(-{\ell}{\sqrt{2\nu}|\mathbf{x}_j-\mathbf{x}_k|}\right)\frac{{(\nu-\frac{1}{2})}!}{(2\nu-1)!} \nonumber \\
\sum_{i=0}^{\nu-\frac{1}{2}}\frac{(\nu-\frac{1}{2}+i)!}{i!(\nu-\frac{1}{2}-i)!}\left({2\ell}{\sqrt{2\nu}|\mathbf{x}_j-\mathbf{x}_k|}\right)^{\nu-\frac{1}{2}-i},
\end{align}
where $\ell>0$ is the scale parameter, and $\nu$ is the half integer smoothness parameter. For example, when $\nu$ is equal to $\frac{1}{2}$, the corresponding kernel function is $\Phi_{\frac{1}{2}}(\mathbf{x}_j,\mathbf{x}_k)=\exp\left(-\ell|\mathbf{x}_j-\mathbf{x}_k|\right)$. The Mat\'ern kernel function is widely applied in practice; see \cite{muyskens2024identifiability} for example. 

For a moment, let us consider the case $d=1$. Without loss of generality, assume $x_1<x_2< \ldots < x_n$. A non-zero function $\phi$ is termed an $s$ degree KP if it can admit the representation as 
\begin{equation} \label{eq: phix}
\phi(x)=\sum_{j=1}^n {A_j \Phi(x, x_j)}
\end{equation}
with the support of $\phi$ being the interval $[x_1, x_n]$. In \eqref{eq: phix}, $A_j$ is the coefficients of $\Phi(x, x_j)$, which can be obtained by solving the following linear systems:
\begin{equation} \label{eq: Aj}
\sum\limits_{j=1}^s A_j x_j^l\exp(\delta c x_j)=0,
\end{equation}
where $l=0,...,\frac{s-3}{2}$, $\delta=\pm1$, $c^2=\frac{2 \nu}{\ell^2}$ ($\nu$ and $\ell$ are the smoothness parameter and scale parameter of Matérn kernel function respectively).
By \eqref{eq: phix}, we have $\mathbf{A} \mathbf{K} = \phi(\mathbf{x})$, and the inverse of the kernel matrix can be computed by
\begin{align} \label{k-1}
    \mathbf{K}^{-1} = \mathbf{A} \phi(\mathbf{x})^{-1},
\end{align}
where $\mathbf{x}=(x_1, x_2,\ldots, x_n)^\top$ is the input vector, and the $(i,j)^{th}$ entry of $\phi(\mathbf{X})$ and $\mathbf{A}$ are $\phi_j(x_i)$ and the $i$ element of $A_j$, respectively. 

It has been shown in \cite{chen2022kernel} that both $\mathbf{A}$ and $\phi(\mathbf{x})$ are sparse banded matrices, hence the computation of $\mathbf{K}^{-1}$ can be fast. Specifically, \cite{chen2022kernel} showed that computing $\mathbf{A}$ and $\phi(\mathbf{x})$ needs only $O((2\nu+2)^3 n)$ computation time and requires storage space of $O((2\nu+2) n)$, where $\nu$ is usually small in most practical applications. LU decomposition \cite{davis2006direct} can be applied for solving $[\phi(\mathbf{x})]^{-1}$, and \cite{chen2022kernel} showed that the total computation time is only $O(n^2)$.

If $d>1$, i.e., the input $\mathbf{X}$ is multi-dimensional, then special structure of the input points $\mathbf{X}=(\mathbf{x_1}, \mathbf{x_2},..., \mathbf{x_d})$ is needed. KP requires the input points to be tensor-like, that is, the input points can be described as the Cartesian combination of multiple one-dimensional point collections. Specifically, $\mathbf{X}^{FG} = \bigtimes_{i=1}^{N} \mathbf{x}^{(i)}$,
where each \(\mathbf{x}^{(i)}\) represents a distinct one-dimensional point collection. With tensor-like input points, $\mathbf{K}$ can be decomposed into Kronecker products as
\begin{equation} \label{K_Kronecker}
    \mathbf{K} = \bigotimes_{i=1}^{d} \mathbf{K_{x_i}},
\end{equation}
where $\mathbf{K_{x_i}} = (\Phi(x_{ij},x_{ik})), 1 \leq i \leq d, 1 \leq j,k \leq n_i$, $n_i$ is the number of points for dimension $i$ and thus the dimension of $\mathbf{K}$ is $n\times n$, where $n = \prod_{i=1}^{d} n_i$.
Then the inverse $\mathbf{K}$ can be directly computed by
\begin{equation}  \label{K-1_Kronecker}
    \mathbf{K}^{-1} = \left[\bigotimes_{i=1}^{d} \mathbf{K_{x_i}}\right]^{-1} = \bigotimes_{i=1}^{d} \mathbf{K_{x_i}^{-1}},
\end{equation}
where each $\mathbf{K_{x_i}^{-1}}$ can be computed using one-dimensional KP. Putting all things together, we obtain
\begin{equation} \label{mul K -1}
    \mathbf{K}^{-1} = \bigotimes_{i=1}^{d} \mathbf{K_{x_i}^{-1}} = \bigotimes_{i=1}^{d} \mathbf{A_{x_i}} \phi(\mathbf{x_i})^{-1},
\end{equation}
whose computation time is still $O(n^2)$.

\section{KP Accelerated PINNs} \label{s:method}

In this section, we introduce the proposed method, KP accelerated PINNs (KP-PINNs). Instead of using $L_2$ loss in \eqref{eq_L2loss}, which was shown to be unstable in some specific PDE systems, we consider tensor Sobolev norm. The tensor Sobolev space \( H_T^\nu(\Omega) \) \cite{hochmuth2000tensor} is defined as the set of functions \( f : \Omega \to \mathbb{R} \) satisfying
\begin{equation} \label{tensor_sobolev_norm}
H_T^\nu(\Omega) = \left\{ f \in L^2(\Omega, \mathbb{R}) \,\middle|\, \norm{f}_{H_T^\nu(\Omega)} < \infty \right\},
\end{equation}
where \( \norm{f}_{H_T^\nu(\Omega)} \) is the norm on \( H_T^\nu(\Omega) \), explicitly defined via the Fourier transform as
\begin{equation} \label{tensor_sobolev_norm_fourier}
\norm{f}_{H_T^\nu(\Omega)}^2 = \int_{\mathbb{R}^d} \left| \tilde{f}(\omega) \right|^2 \prod_{i=1}^d \left( 1 + |\omega_i|^2 \right)^\nu d\omega,
\end{equation}
where \( \tilde{f}(\omega) \) represents the Fourier transform of \( f \). We define the loss function with tensor Sobolev norm as
\begin{align}\label{eq_sobloss}
    {\rm Loss}_{H_T^\nu} = \norm{ \mathcal{L} u_\text{NN} - f} _{H_T^\nu(\Omega)}^{2} + \lambda \norm{\mathcal{B} u_\text{NN} - g } _{H_T^\nu(\partial \Omega)}^{2},
\end{align}
where $\lambda>0$ is a tuning parameter that balances the weights between the difference between the space and its boundary, and $u_{NN}$ is the predicted value obtained from the neural network. However, directly computing the tensor Sobolev norm may be time-consuming and inaccurate. By the equivalence between the RKHS norm and tensor Sobolev norm, we modify the loss function \eqref{eq_sobloss} to
\begin{align}\label{eq_rkhsloss}
    {\rm Loss}_{\mathcal{N}_{\Phi_T^\nu}} = \norm{\mathcal{L} u_\text{NN} - f} _{\mathcal{N}_{\Phi_T^\nu(\Omega)}}^{2} + \lambda \norm{\mathcal{B} u_{\text{NN}} - g} _{\mathcal{N}_{\Phi_T^\nu(\partial\Omega)}}^{2},
\end{align}
where $\mathcal{N}_{\Phi_T^\nu}$ is tensor RKHS, i.e. the RKHS generated by $\Phi_T^\nu(x) = \prod_{i=1}^d \Phi^\nu(x_i)$. Tensor RKHS and Sobolev space are widely considered for  complexity reduction in high-dimensional spaces \cite{ding2020high,kuhn2015approximation,dung2021deep}.

Computing the loss function in \eqref{eq_rkhsloss} can be efficient via KP. Specifically, let
\begin{equation} \label{eq: h}
h_{1\text{NN}} = \mathcal{L} u_\text{NN} - f, h_{2\text{NN}} = \mathcal{B} u_\text{NN} - g,
\end{equation}
and we consider $\norm{h_{1\text{NN}}} _{\mathcal{N}_{\Phi_T^\nu(\Omega)}}^{2} = \norm{\mathcal{L} u_\text{NN} - f} _{\mathcal{N}_{\Phi_T^\nu(\Omega)}}^{2}$ first. Based on the Theorem 11.23 in \cite{wendland2004scattered}, $\norm{h_{1\text{NN}}} _{\mathcal{N}_{\Phi_T^\nu(\Omega)}}^{2}$ can be approximated by 
\begin{equation} \label{eq: RKHS h}
\mathbf{y}_{h1}^T \mathbf{K}^{-1} \mathbf{y}_{h1},
\end{equation}
where $\mathbf{y}_{h1} = (h_{1\text{NN}}(\mathbf{x^1}), ... , h_{1\text{NN}}(\mathbf{x^n}))$ and ${\mathbf{x^1}, ..., \mathbf{x^n}} \in \mathbb{R}^{d}$ means all the training data. $\mathbf{K}$ is as in \eqref{K_Kronecker}. The other term $\norm{h_{2\text{NN}}} _{\mathcal{N}_{\Phi_T^\nu(\partial\Omega)}}^{2}$ can be approximated in a similar manner.

With \eqref{k-1}, \eqref{mul K -1} and approximation \eqref{eq: RKHS h}, the loss function \eqref{eq_rkhsloss} can be approximated by:
\begin{align} \label{loss KP}
    {\rm Loss}_{\mathcal{N}_{\Phi_T^\nu}} \approx 
    \mathbf{y}_{h1}^T \mathbf{A}_{h1} \phi(\mathbf{x})_{h1}^{-1} \mathbf{y}_{h1} + \lambda \mathbf{y}_{h2}^T \mathbf{A}_{h2} \phi(\mathbf{x})_{h2}^{-1} \mathbf{y}_{h2}.
\end{align}
By KP method introduced in Section \ref{s:prel.4}, ${\rm Loss}_{\mathcal{N}_{\Phi_T^\nu}}$ can be easily computed.Some well-established optimization algorithms, such as Adam or L-BFGS, can be used to minimize the loss effectively. The KP-PINNs algorithm is summarized in Algorithm \ref{alg:multi}. For the inverse problem, comparing \eqref{eq: MSE forward} and \eqref{eq: MSE inverse}, it can be seen that the loss function of the inverse problem has more ${\rm Loss}_{\mathcal{D}}(u)$ part. For this part, the computation of the RKHS norm is similar to the analysis above. In this way, the KP-PINNs algorithm can effectively solve both forward and inverse problems.

\begin{algorithm}[tb]
    \caption{KP-PINNs}
    \label{alg:multi}
    \textbf{Input}: PDE systems \eqref{eq: normalized pde}, iterations $n_\text{iter}$, known points\\
    \textbf{Output}: Approximated solution $\hat{u}(\mathbf{x})$ and equations' parameters (if inverse problem)
    \begin{algorithmic}[1] 
        \STATE Initialize neural network parameters.
        \WHILE{$i < n_\text{iter}$}
        \STATE Forward pass and compute the predicted value.
        \STATE Compute derivatives using automatic differentiation.
        \STATE Compute $\mathbf{A}$ and $\phi(\mathbf{x})$ by \eqref{eq: Aj} and  \eqref{eq: phix}.
        \STATE Compute the loss of KP-PINNs based on \eqref{loss KP}.
        \STATE Update parameters.
        \ENDWHILE
    \end{algorithmic}
\end{algorithm}

\section{Theoretical Results} \label{s:theor}

In this section, we analyze the stability of PDEs solved with the loss function \eqref{eq_sobloss}, or equivalently \eqref{eq_rkhsloss}. The stability of a PDE system is defined as follows \cite{wang20222}.
\begin{definition}[Stability of PDEs \cite{wang20222}]
\label{def:stability}
Suppose $Z_1, Z_2$, and $Z_3$ are three Banach spaces. For the exact solution $u^*(\mathbf{x})$ and the predicted solution $u$, if for $\norm{\mathcal{L}u - f}_{Z_1}, \norm{\mathcal{B}u - g}_{Z_2} \to 0$, it has $\norm{u^* - u}_{Z_3} = O(\norm{\mathcal{L}u - f}_{Z_1} + \norm{\mathcal{B}u - g}_{Z_2})$, then the PDEs \eqref{eq: normalized pde} is $(Z_1, Z_2, Z_3)$-stable. 
\end{definition}
Stability defined in Definition \ref{def:stability} provides a theoretical guarantee for solving a PDE. Specifically, if the loss function is $\norm{\mathcal{L}u - f}_{Z_1} + \norm{\mathcal{B}u - g}_{Z_2}$, then minimizing the loss function can ensure the convergence of the solution under $Z_3$-norm. In this section, we consider two classes of PDEs: second-order linear elliptic equations and Hamilton-Jacobi-Bellman (HJB) equations. We start with the second-order linear elliptic equations satisfying the following assumption. While in the previous section, we considered the parameters of the PDE as constants, here we extend the analysis to a more general case, where the parameters are functions. This allows us to establish the theoretical results in a broader context. 
\begin{assumption}\label{A1}
    Suppose in \eqref{eq: normalized pde}, the operator $\mathcal{L} = \sum_{i,j=1}^d a_{i,j}(\mathbf{x}) \partial_{x_i} \partial_{x_j} + \sum_{i=1}^d b_{i}(\mathbf{x}) \partial_{x_i} + c(x)$ and $\mathcal{B}$ is the identity operator. The functions $a_{i,j}\in C^2$ satisfy the uniformly elliptic condition. In addition, the only solution with zero input data is the zero solution.
\end{assumption}   
\begin{lemma}[Theorem 2.1 in \cite{bramble1970rayleigh}]\label{Lemma elliptic bound}
    In addition to Assumption \eqref{A1}, assume that $\mathcal{L}$ is with $C^\infty$ coefficients, defined on $C^\infty$ bounded domain $\Omega$ on $\mathbb{R}^d$. For any real number $l$, $\norm{u}_{H^l(\Omega)} \leq C\left(\norm{\mathcal{L} u}_{H^{l-2}(\Omega)}+\norm{u}_{H^{l-\frac{1}{2}}(\partial \Omega)}\right)$ for all $u \in C^\infty(\bar{\Omega})$, where $\bar{\Omega}$ is a closure of $\Omega$ and C is independent of $u$.
\end{lemma}
\begin{theorem}[Stability of second-order linear elliptic equation]\label{thm:elliptic_RKHS}
Suppose the RKHS $\mathcal{N}_{\Phi_T^1}$ coincides with the space $H^1_T$. For any $C^{\infty}$ bounded domain $\Omega \subseteq \mathbb{R}^d$, if Assumption \ref{A1} is satisfied and all the coefficient functions are in $C^\infty$, then the equation 
    \begin{equation*}
        \left\{\begin{array}{ll}
        \mathcal{L}_{\theta} u(\mathbf{x}) = h_{\theta}(\mathbf{x}), & \mathbf{x} \in \Omega \subset \mathbb{R}^{d}, \\
        \mathcal{B}_{\theta} u(\mathbf{x}) = g_{\theta}(\mathbf{x}), & \mathbf{x} \in \partial \Omega
        \end{array}\right.
    \end{equation*}
    is $(\mathcal{N}_{\Phi^1_T} (\Omega),\mathcal{N}_{\Phi^1_T} (\partial \Omega), H^1(\Omega))$-stable.
\end{theorem}
\begin{proof}
Sets $l=1$, by Lemma \ref{Lemma elliptic bound}, we have 
\begin{align*}
    \norm{u}_{H^1(\Omega)} &\leq C_1 (\norm{\mathcal{L} u }_{H^{-1}(\Omega)} + \norm{u}_{H^{\frac{1}{2}}(\partial \Omega)}) \\
    &\leq C_2( \norm{\mathcal{L} u }_{L^{2}(\Omega)} + \norm{\mathcal{B} u }_{H^{1}(\partial \Omega)}) \\
    &\leq C_3( \norm{\mathcal{L} u }_{H^1(\Omega)} + \norm{\mathcal{B} u }_{H^{1}(\partial \Omega)}) \\
    & \leq C_4 ( \norm{\mathcal{L} u }_{H^1_T(\Omega)} + \norm{\mathcal{B} u }_{H^{1}_T(\partial \Omega)}).
\end{align*}
Since the RKHS $\mathcal{N}_{\Phi_T^1}$ coincides with the space $H^1_T$, we have that $\norm{u}_{H^1(\Omega)} \leq C_4( \norm{\mathcal{L} u }_{\mathcal{N}_{\Phi_T^1}(\Omega)} + \norm{\mathcal{B} u }_{\mathcal{N}_{\Phi_T^1}(\partial \Omega)})$.
\end{proof}
The proof of Theorem \ref{thm:elliptic_RKHS} is a natural result derived from the classical theory of second-order linear elliptic equations \cite{shin2023error}.  
It shows that the numerical solution for the second-order linear elliptic equation is stable.
Next, we show the stability of the HJB equation \eqref{eq: reformulated HJB}.
\begin{lemma}[Stability of HJB equation - $L_\infty$ norm \cite{wang20222}]
\label{thm:hjb_Lp}
For $p,q \geq 1$, let $r_0 = \frac{(d+2)q}{d+q}$. Suppose the subsequent inequalities are valid for $p,q$ and $r_0$:
\begin{equation*}
\begin{aligned}
    p \geq \max\left\{2,\left(1-\frac1{\check{c}}\right)d\right\},
    q > \frac{(\check{c}-1)d^2}{(2-\check{c})d+2},
    \frac1{r_0} \geq \frac1p-\frac1d,
\end{aligned}
\end{equation*}
where $\check{c}=\max_{1\leq i\leq d}c_i$ in \eqref{eq: reformulated HJB}. For any $r\in[1,r_0)$, any bounded open set $\Omega \subset \mathbb{R}^d\times[0,T]$, and for a large $R$ such that $\Omega \subset B_R \times [0,T]$, where $B_R$ is a ball with radius $R$, \eqref{eq: reformulated HJB} maintains $(L_p(\Omega),L_q(B_R),W^{1,r}(\Omega))$-stability when $\check{c} \leq 2$.
\end{lemma}
\begin{theorem}[Stability of HJB equation]
\label{thm:hjb_RKHS}
For $p,q \geq 1$, let $r_0 = \frac{(d+2)q}{d+q}$. Suppose the subsequent inequalities are valid for $p,q$ and $r_0$:
\begin{equation*}
p\geq\max\left\{2,\left(1-\frac1{\check{c}}\right)d\right\}, 
q>\frac{(\check{c}-1)d^2}{(2-\check{c})d+2}, 
\frac1{r_0}\geq\frac1p-\frac1d,
\end{equation*}
where $\check{c}=\max_{1\leq i\leq d}c_i$ in \eqref{eq: reformulated HJB}. Then, for any $r\in[1,r_0)$, $\nu_1,\nu_2 \geq 1$ and any bounded open set $\Omega\subset\mathbb{R}^d\times[0,T]$, \eqref{eq: reformulated HJB} maintains $(\mathcal{N}_{\Phi^{\nu_1}_T}(\mathbb{R}^d\times[0,T]),\mathcal{N}_{\Phi^{\nu_2}_T}(\mathbb{R}^d),W^{1,r}(\Omega))$-stable when $\check{c} \leq 2$.
\end{theorem}
\begin{proof}
First of all, according to the Sobolev embedding theorem \cite{adams2003sobolev,wendland2004scattered,ding2019bdrygp}, when bounded set $\Omega \subset \mathbb{R}^d$, $p<\infty$, $\forall f \in H^1$, $\exists c_1,c_2,c_3,c_4,c_5>0$,  it has the norm inequality
\begin{equation} \label{eq: Norm inequality}
\begin{split}
    c_1 \norm{f}_{L_p(\Omega)} &\le 
c_2 \norm{f}_{L_\infty(\Omega)} \le
c_3 \norm{f}_{H^1(\Omega)}\\
&\le
c_4 \norm{f}_{H_T^1(\Omega)} \le
c_5 \norm{f}_{\Phi_T^1(\Omega)}.
\end{split}
\end{equation}
If $\norm{\mathcal{L}_\text{HJB}u(\mathbf{x}) - \varphi(\mathbf{x})}_{H^{\nu_1}}, \norm{\mathcal{B}_\text{HJB}u(\mathbf{x}) - g(\mathbf{x})}_{H^{\nu_2}} \to 0$, known by \eqref{eq: Norm inequality}, $\norm{\mathcal{L}_\text{HJB}u(\mathbf{x}) - \varphi(\mathbf{x})}_{L_p}, \norm{\mathcal{B}_\text{HJB}u(\mathbf{x}) - g(\mathbf{x})}_{L_q} \to 0$. By Lemma \ref{thm:hjb_Lp}, \eqref{eq: reformulated HJB} is stable. It can be represented as $(L_p(\Omega),L_q(B_R),W^{1,r}(\omega))$-stable, and thus $(H^{\nu_1}(\mathbb{R}^d\times[0,T]), H^{\nu_2}(\mathbb{R}^d), W^{1,r}(\Omega))$-stable. Additionally, the equation is also $(\mathcal{N}_{\Phi_T^1}(\mathbb{R}^d\times[0,T]), \mathcal{N}_{\Phi_T^1}(\mathbb{R}^d), W^{1,r}(\Omega))$-stable, thus $(\mathcal{N}_{\Phi_T^{\nu_1}}(\mathbb{R}^d\times[0,T]), \mathcal{N}_{\Phi_T^{\nu_2}}(\mathbb{R}^d), W^{1,r}(\Omega))$-stable.
\end{proof} 
Theorem \ref{thm:hjb_RKHS} is a modified version of Theorem 4.3 in \cite{wang20222}.
It demonstrates that when the dimension of the state function $d$ is large, selecting an appropriate RKHS ensures the stability of the equations.

\section{Experiments} \label{s:exper}
In this section, we apply KP-PINNs to four representative differential equations: an ODE - Stiff equation, a second-order linear PDE - Helmholtz equation, a HJB equation - LQG equation, and a nonlinear PDE - NS equation. For each differential equation, both the forward and inverse problems are addressed. We compare the KP-PINNs algorithm with several baseline PINNs approaches including $L_2$-PINNs, RKHS-PINNs, and Sobolev-PINNs, as summarized in Table \ref{comparing algorithms}. The implementation details and source code are available at: \url{https://github.com/SiyuanYang-sy/KP-PINNs}.

We consider the accuracy and computational time, where the accuracy is measured by the relative $L_2$ error defined by
\begin{align}
\norm{e}_{\text{relative $L_2$}} = \left( \frac{ \sum_{i=1}^{N_{\text{test}}} |\hat{u}(\mathbf{x}_i) - u(\mathbf{x}_i)|^2 }{ \sum_{i=1}^{N_{\text{test}}} |u(\mathbf{x}_i)|^2 } \right)^{\frac{1}{2}},
\end{align}
where $u(\mathbf{x}_i)$ represents the true solution at $\mathbf{x}_i$, $\hat{u}(\mathbf{x}_i)$ denotes the prediction obtained from the algorithm, and $N_{\text{test}}$ is the size of the test set.
Also, the standard error (SE) is given to reflect the variability across multiple experiments.

\begin{table} [b] 
\renewcommand{\arraystretch}{1.15}
\centering
\begin{tabular}{lc}
\toprule
\multicolumn{1}{c}{Algorithm}               & Loss function  \\ 
\midrule
KP-PINNs ($\mathcal{N}_{\Phi_T^{\nu}}$ norm)   & \eqref{loss KP}, $\nu$ equals to $\frac{1}{2}, \frac{3}{2}, \frac{5}{2}$    \\
$L_2$-PINNs ($L_2$ norm)                          & $\frac{1}{n} \sum_{i=1}^{n} (\mathbf{y}_i - \mathbf{\hat{y}}_i)^2$   \\
RKHS-PINNs ($\mathcal{N}_{\Phi_T^{\nu}}$ norm)  & \eqref{eq: RKHS h}, $\nu$ equals to $\frac{1}{2}, \frac{3}{2}, \frac{5}{2}$     \\
Sobolev-PINNs ($H_T^\nu$ norm)                  & \eqref{tensor_sobolev_norm_fourier}, $\nu$ equals to $1, 2, 3$      \\
\bottomrule
\end{tabular}
\caption{KP-PINNs and the comparing algorithms.}
\label{comparing algorithms}
\end{table}

\subsection{Stiff Equation} \label{s:exper.2}
The Stiff equation \cite{wanner1996solving} exhibits steep gradients caused by fast dynamics, and is commonly treated with time-scale decomposition to reduce numerical and reaction-related errors.
It is defined by
\\
\begin{equation}\label{eq_stiff_eq}
u'(t) - \lambda u(t) = e^{-t}, t \in [0, 5], u(0) = \mu,
\end{equation}
where $\lambda$ and $\mu$ are two parameters. Here we set $\lambda = -2.0$ and $\mu = 2.0$. The analytic solution to \eqref{eq_stiff_eq} is
\begin{equation}
u(t)=\left(\mu+\frac1{1+\lambda}\right)e^{\lambda t}-\frac{e^{-t}}{1+\lambda}.
\end{equation}

\begin{table} [b] 
\renewcommand{\arraystretch}{1.15}
\centering
\resizebox{0.48\textwidth}{!}{%
\begin{tabular}{lcccccc}
\toprule
\multicolumn{1}{c}{Algorithm (-PINNs)} & Forward & SE & Inverse & SE & Inv $\lambda ($-2.0$)$ \\
\midrule
KP ($\nu=\frac{1}{2}$)          & 2.23e-04 & 4.90e-05 & 3.95e-04 & 1.37e-04 & -1.99983 \\
KP ($\nu=\frac{3}{2}$)          & 5.57e-04 & 3.83e-04 & 5.62e-04 & 2.79e-04 & -2.00015 \\
KP ($\nu=\frac{5}{2}$)          & 1.55e-04 & 6.64e-05 & 3.25e-04 & 1.43e-04 & -1.99991 \\
$L_2$                           & 3.31e-04 & \textbf{4.13e-05} & 3.32e-04 & \textbf{1.34e-04} & \textbf{-1.99995} \\
RKHS ($\nu=\frac{1}{2}$)        & 2.90e-04 & 8.19e-05 & 3.93e-04 & 1.35e-04 & -1.99983 \\
RKHS ($\nu=\frac{3}{2}$)        & 5.83e-04 & 3.70e-04 & 3.75e-04 & 3.25e-04 & -1.99984 \\
RKHS ($\nu=\frac{5}{2}$)        & \textbf{1.29e-04} & 5.09e-05 & \textbf{3.24e-04} & 1.41e-04 & -1.99991 \\
Sobolev ($\nu=1$)               & 5.62e-02 & 4.57e-02 & 4.38e-01 & 8.75e-03 & -1.05663 \\
Sobolev ($\nu=2$)               & 3.58e-01 & 1.38e-02 & 4.12e-01 & 1.94e-02 & -1.29371 \\
Sobolev ($\nu=3$)               & 3.56e-01 & 8.57e-03 & 4.29e-01 & 2.60e-02 & -1.64380 \\
\bottomrule
\end{tabular}%
}
\caption{Experimental results of Stiff equation. These six columns represent algorithm types, forward error results, forward SE results, inverse error results, inverse SE results and the parameter to be estimated in the inverse problems. The following tables are also represented in this way.}
\label{exp-stiff}
\end{table}

There is only one initial point so $N_{\mathcal{B}}=1$, and $N_{\mathcal{L}}=50$ in \eqref{eq: MSE forward} in the forward problem. In the inverse problem, we assume that $\lambda$ is unknown and set $N_{\mathcal{B}}+N_{\mathcal{L}}=50$.
In both cases $N_{\text{test}}$ is equal to $2000$.
Table \ref{exp-stiff} provides the average results of three independent runs. KP-PINNs and RKHS-PINNs yield similar results due to the same loss function form, but differ in the computation of $\mathbf{K}^{-1}$, leading to the run-time gap observed in Table \ref{k-time}, where KP-PINNs demonstrate significantly faster computation across all four equations. The Sobolev-PINNs algorithm uses automatic differentiation to compute first- and higher-order derivatives, leading to error accumulation and difficulty in obtaining accurate results. 

\begin{figure} [t]  
    \centering
    \includegraphics[width=\linewidth]{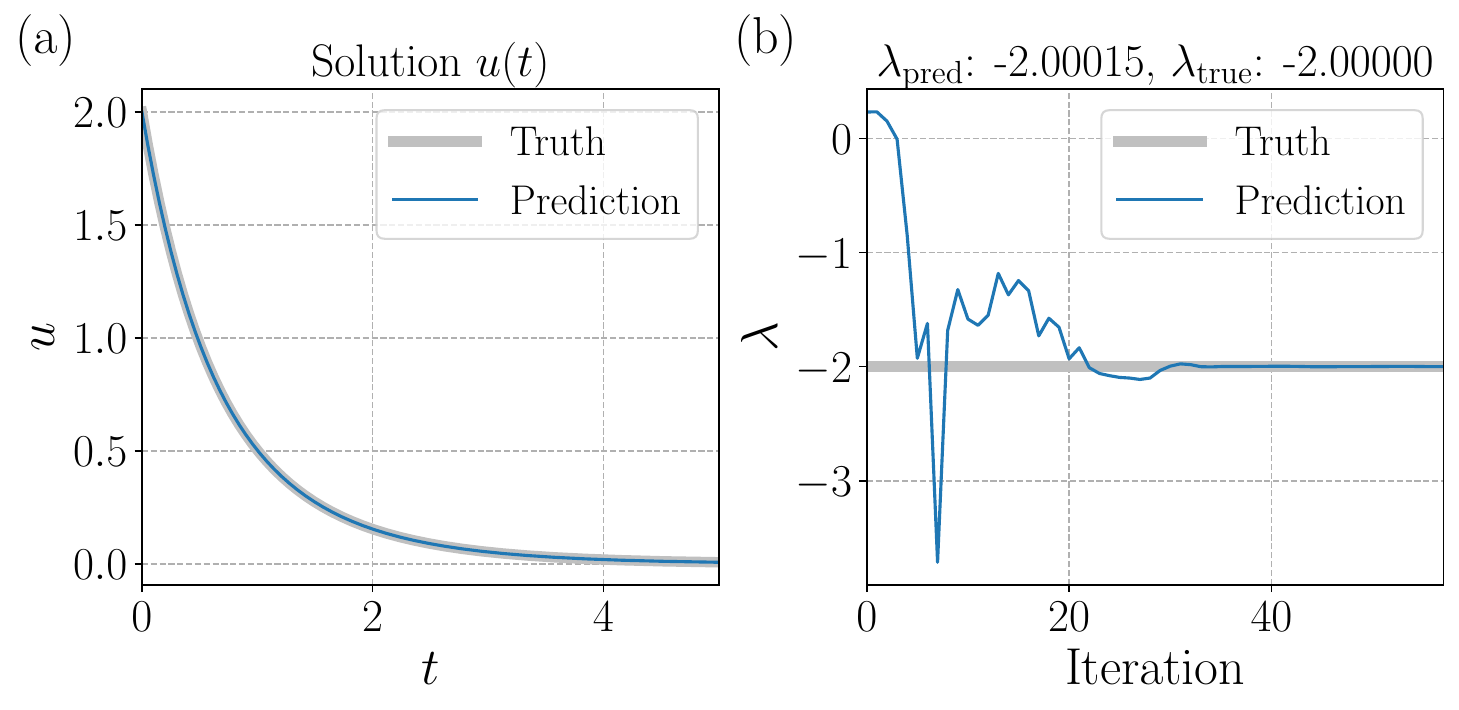}
    \caption{KP-PINNs ($\nu=\frac{3}{2}$) inverse results of Stiff equation. (a) Predicted $u(t)$. (b) Predicted $\lambda$.}
    \label{Stiff-KP32}
\end{figure}


\subsection{Helmholtz Equation} \label{s:exper.3}
The Helmholtz equation is a second-order elliptic PDE commonly encountered in electromagnetism that describes the spatial behavior of electromagnetic wave propagation. It is defined by
\\
\begin{equation} \label{Helmholtz equation}
\begin{cases}
    \Delta u + k^2u = p, (x,y)\in \Omega = [-1,1]^2,    \\ 
    u(x,y)=0, (x,y)\in \partial\Omega,
\end{cases}
\end{equation}
\noindent where $\Delta u = \frac{\partial^2u}{\partial x^2}+\frac{\partial^2u}{\partial y^2}$, $k$ is the wave number, $p$ is the source term (excitation term) \cite{li2022dwnn} defined by $p(x,y) = (x+y) \sin(\pi x) \sin(\pi y) - 2\pi^2 (x+y) \sin(\pi x) \sin(\pi y) + 2\pi \cos(\pi y) \sin(\pi x) + 2\pi \cos(\pi x) \sin(\pi y)$. Here, we set $k=1.0$. The analytic solution for $k=1.0$ \cite{jagtap2020adaptive} is
\begin{equation}
    u(x,y)=(x+y)\sin(\pi x)\sin(\pi y).
\end{equation}

\begin{table} [b] 
\renewcommand{\arraystretch}{1.15}
\centering
\resizebox{0.48\textwidth}{!}{%
\begin{tabular}{lcccccc}
\toprule
\multicolumn{1}{c}{Algorithm (-PINNs)} & Forward & Standard & Inverse & Standard & Inv $k$ ($1.0$) \\
\midrule
KP ($\nu=\frac{1}{2}$)          & 1.13e-03 & 1.97e-04 & 4.88e-03 & 1.10e-03 & 0.98262 \\
KP ($\nu=\frac{3}{2}$)          & \textbf{9.13e-04} & 1.91e-04 & 1.26e-02 & 7.35e-03 & 0.83391 \\
KP ($\nu=\frac{5}{2}$)          & 1.15e-03 & 1.91e-04 & 2.75e-02 & 7.20e-03 & 0.51944 \\
$L_2$                           & 2.99e-03 & 3.06e-04 & 2.36e-01 & 2.05e-01 & 4.37582 \\
RKHS ($\nu=\frac{1}{2}$)        & 2.10e-03 & 9.61e-04 & \textbf{3.37e-03} & \textbf{6.26e-04} & \textbf{0.99927} \\
RKHS ($\nu=\frac{3}{2}$)        & 1.02e-03 & \textbf{1.23e-04} & 1.68e-02 & 7.58e-03 & 0.77236 \\
RKHS ($\nu=\frac{5}{2}$)        & 1.10e-03 & 1.65e-04 & 1.88e-02 & 7.34e-03 & 0.69293 \\
Sobolev ($\nu=1$)               & 1.79e+00 & 1.01e+00 & 7.41e-01 & 4.82e-01 & 0.00475 \\
Sobolev ($\nu=2$)               & 8.95e-01 & 2.09e-02 & 3.85e-01 & 1.64e-02 & 0.00221 \\
Sobolev ($\nu=3$)               & 1.31e+01 & 8.85e+00 & 8.20e-01 & 1.00e-01 & 0.01159 \\
\bottomrule
\end{tabular}%
}
\caption{Experimental results of Helmholtz equation.}
\label{exp-helmholtz}
\end{table}




\begin{figure}
    \centering
    \includegraphics[width=\linewidth]{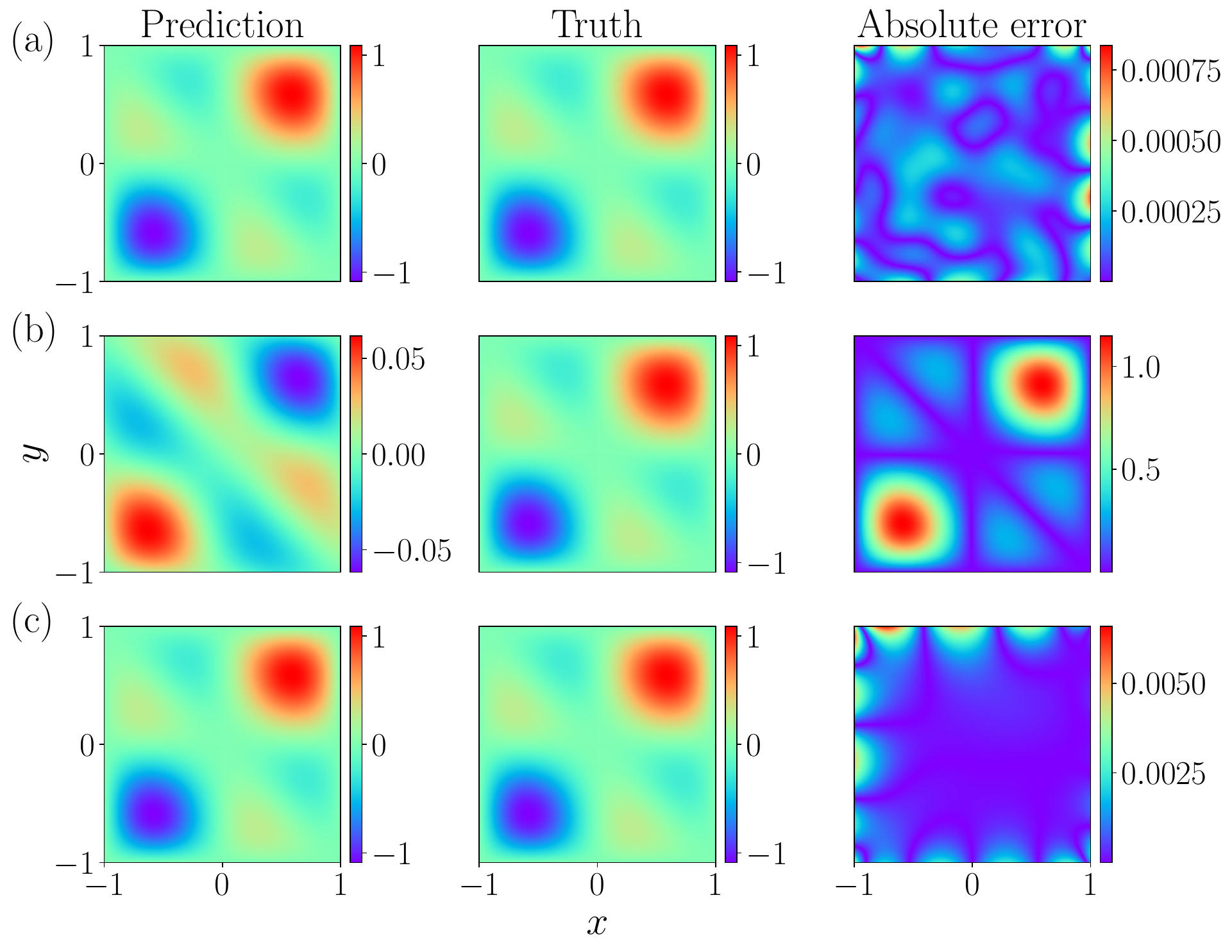}
    \caption{Comparison of predicted $u(x, y)$ for Helmholtz equation using different methods. (a) KP-PINNs ($\nu=\frac{3}{2}$) forward results. (b) $L_2$-PINNs inverse results. (c) KP-PINNs ($\nu=\frac{1}{2}$) inverse results.}
    \label{fig:combined-helmholtz}
\end{figure}

In this experiment, $N_{\mathcal{B}}=100$ and $N_{\mathcal{L}}=100\times 100$ in \eqref{eq: MSE forward}. For the inverse problem, parameter $k$ is unknown and $N_{\mathcal{B}}+N_{\mathcal{L}}=30\times 30$ in \eqref{eq: MSE forward}. In both cases $N_{\text{test}}$ is equal to $500\times 600$. 
Table \ref{exp-helmholtz} gives the average results of five times. It can be seen that the first three algorithms can get the solutions well. For inverse problems, only KP-PINNs ($\nu=\frac{1}{2}$) and RKHS-PINNs ($\nu=\frac{1}{2}$) can predict the $k$ relatively correctly, while the $L_2$-PINNs fails to work. Figure \ref{fig:combined-helmholtz} (b), Figure \ref{fig:combined-helmholtz} (c) and Figure \ref{Helmholtz-LQG-parameters} (a) show the inverse problem results of $L_2$-PINNs and KP-PINNs ($\nu=\frac{1}{2}$). The parameter $\nu$ determines the degree of smoothness for the kernel function, and a smaller value of $\nu$ is suitable for less smooth conditions. It can be seen that a smaller $\nu$ (such as $\nu=\frac{1}{2}$) can get better results.

\subsection{Linear Quadratic Gaussian (LQG) Equation} \label{s:exper.4}
One application of the HJB equation is LQG control problem. 
The form of LQG \cite{han2018solving} in $d$ dimensions are
\begin{equation}
\begin{cases}
\partial_tu(\mathbf{x},t)+\Delta u(\mathbf{x},t)-\mu\norm{\nabla_\mathbf{x}u(\mathbf{x},t)}^2=0, \\
\mathbf{x}\in\mathbb{R}^d,t\in[0,T], \\
u(\mathbf{x},T)=g(\mathbf{x}), \mathbf{x}\in\mathbb{R}^d,
\end{cases}
\end{equation}
with the solution  
\begin{equation}
u(\mathbf{x},t)= -\frac1\mu \ln\left(\int_{\mathbb{R}^d} (2\pi)^{\frac{d}{2}} \mathrm{e}^{-\frac{\|y\|^2}{2}} \mathrm{e}^{-\mu g(\mathbf{x}-\sqrt{2(T-t)}y)}\mathrm{d}y \right),
\end{equation}
where $g(\mathbf{x})=\ln\left(\frac{1+\norm{\mathbf{x}}^2}2\right)$ is the terminal cost function.  


\begin{figure}
    \centering
    \includegraphics[width=\linewidth]{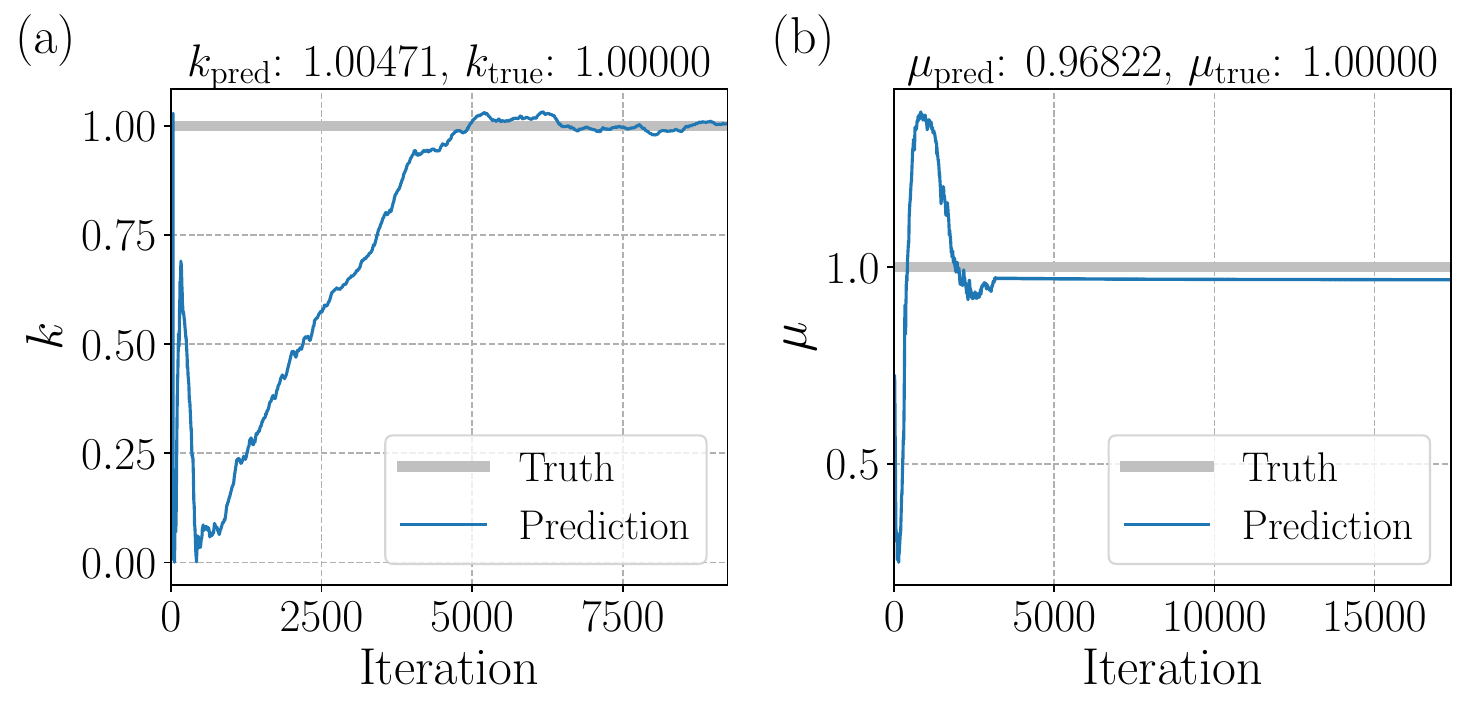}
    \caption{KP-PINNs ($\nu=\frac{1}{2}$) inverse results. (a) Helmholtz equation - predicted $k$. (b) LQG equation - predicted $\mu$.}
    \label{Helmholtz-LQG-parameters}
\end{figure}

\begin{table} [b] 
\renewcommand{\arraystretch}{1.15}
\centering
\resizebox{0.48\textwidth}{!}{%
\begin{tabular}{lcccccc}
\toprule
\multicolumn{1}{c}{Algorithm (-PINNs)} & Forward & Standard & Inverse & Standard & Inv $\mu$ ($1.0$) \\
\midrule
KP ($\nu=\frac{1}{2}$)          & \textbf{6.70e-02} & 2.43e-02 & \textbf{6.90e-03} & 8.84e-04 & \textbf{0.93244} \\
KP ($\nu=\frac{3}{2}$)          & 2.26e-01 & 8.56e-02 & 1.81e-02 & 1.45e-03 & 1.26715 \\
KP ($\nu=\frac{5}{2}$)          & 3.46e-01 & 1.89e-02 & 1.88e-02 & 2.29e-03 & 1.27243 \\
$L_2$                           & 2.39e-01 & 9.80e-02 & 1.77e-02 & 1.36e-03 & 1.20942 \\
RKHS ($\nu=\frac{1}{2}$)        & 1.71e-01 & 6.90e-02 & 8.25e-03 & 2.28e-03 & 0.93072 \\
RKHS ($\nu=\frac{3}{2}$)        & 3.21e-01 & 8.89e-03 & 1.88e-02 & 2.17e-03 & 1.27210 \\
RKHS ($\nu=\frac{5}{2}$)        & 3.39e-01 & 5.54e-03 & 1.89e-02 & 1.94e-03 & 1.25183 \\
Sobolev ($\nu=1$)               & 1.42e+00 & 3.43e-02 & 6.05e-01 & 9.27e-02 & -1.21866 \\
Sobolev ($\nu=2$)               & 1.43e+00 & \textbf{3.37e-04} & 7.12e-03 & \textbf{9.39e-05} & -129.13841 \\
Sobolev ($\nu=3$)               & 7.05e-02 & 3.33e-02 & 1.51e-02 & 3.12e-03 & 1.16909 \\
\bottomrule
\end{tabular}%
}
\caption{Experimental results of LQG equation.}
\label{exp-lqg}
\end{table}

We set $d=2, \mu = 1.0, T = 1.0, N_{\mathcal{B}}=2000$ and $N_{\mathcal{L}}=30\times 20\times 10$ in \eqref{eq: MSE forward} in the forward problem. For the inverse problem, $\mu$ is unknown and $N_{\mathcal{B}}+N_{\mathcal{L}}=10\times 10\times 10$. In both cases, $N_{\text{test}}$ is equal to $100\times 80\times 50$.
Table \ref{exp-lqg} gives the average results of three times experiments. For the inverse problem, KP-PINNs with $\nu=\frac{1}{2}$ can obtain the closest estimation among all algorithms. Figure \ref{Helmholtz-LQG-parameters} (b) gives a one-case inverse problem results of KP-PINNs ($\nu=\frac{1}{2}$).

\subsection{Navier-Stokes (NS) Equation} \label{s:exper.5}
The widely-used NS equations are fundamental PDEs that describe the motion of viscous and incompressible fluids. In this experiment, the two-dimensional incompressible NS equations are considered to model the evolution of fluid flow.
The governing equation for the vorticity $\omega(t, x, y)$ is
\begin{equation}
\begin{split}
\omega_t + u \omega_x + v \omega_y &= \mu (\omega_{xx} + \omega_{yy}), \\
t \in [0, T], x \in [x_{\text{start}}, &x_{\text{end}}], y \in [y_{\text{start}}, y_{\text{end}}],
\end{split}
\end{equation}
where $\mu$ is the viscosity coefficient. In our experiment, we set $\mu = 0.01$, $T=1.9$, $x_{\text{start}}=1.0$, $x_{\text{end}}=8.0$, $y_{\text{start}}=-2.0$ and $y_{\text{end}}=2.0$. The initial condition is the situation of $t=0$. The boundary conditions are the situation of $x=x_{\text{start}}$, $x=x_{\text{end}}$, $y=y_{\text{start}}$ and $y=y_{\text{end}}$.
The vorticity is defined in terms of the stream function $\psi(t, x, y)$ by the Poisson equation $\omega = -(\psi_{xx} + \psi_{yy})$.
The velocity components $(u, v)$ can be recovered from the stream function as $u = \psi_y, \quad v = -\psi_x$.
The two components of the velocity field data $u(t,x,y)$ and $v(t,x,y)$ are obtained from \cite{raissi2019physics}, where the numerical solution is performed.
We approximate the stream function $\psi$ using a neural network. The velocity field and vorticity can be derived by automatic differentiation.

\begin{table} [!bt] 
\renewcommand{\arraystretch}{1.15}
\centering
\resizebox{0.48\textwidth}{!}{%
\begin{tabular}{lcccccc}
\toprule
\multicolumn{1}{c}{Algorithm (-PINNs)} & Forward & Standard & Inverse & Standard & Inv $\mu$ ($0.01$) \\
\midrule
KP ($\nu=\frac{1}{2}$)          & \textbf{1.66e-01} & 9.90e-03 & \textbf{3.61e-02} & 1.75e-02 & \textbf{0.01176} \\
$L_2$                           & 3.97e-01 & 1.18e-01 & 1.40e-01 & 5.46e-02 & 0.01233 \\
RKHS ($\nu=\frac{1}{2}$)        & 1.73e-01 & \textbf{6.66e-03} & 4.87e-02 & \textbf{1.07e-02} & 0.01201 \\
\bottomrule
\end{tabular}%
}
\caption{Experimental results of NS equation. As demonstrated in the experiments above, a smaller value of $\nu$ is more suitable for less smooth conditions. For the NS equation, KP-PINNs ($\nu=\frac{1}{2}$) and RKHS-PINNs ($\nu=\frac{1}{2}$) are employed to evaluate the performance.}
\label{exp-ns}
\end{table}
We set $N_{\mathcal{B}}=5000$ and $N_{\mathcal{L}}=30\times 20\times 10$ in \eqref{eq: MSE forward}. For the inverse problem, $\mu$ is unknown and $N_{\mathcal{B}}+N_{\mathcal{L}}=10\times 10\times 10$. In both cases, $N_{\text{test}} = 98\times 48\times 20$. 
Table \ref{exp-ns} shows the averaged results over three independent runs. It can be seen that KP-PINNs, $L_2$-PINNs and RKHS-PINNs can solve the equation in general. Among them, KP-PINNs and RKHS-PINNs consistently demonstrate superior performance in terms of both accuracy and robustness. For the inverse problem, KP-PINNs($\nu=\frac{1}{2}$) can obtain the closest estimation.  

\begin{table} [!bt] 
\renewcommand{\arraystretch}{1.15}
\centering
\resizebox{0.48\textwidth}{!}{%
\begin{tabular}{lcccc}
\toprule
\multicolumn{1}{c}{Algorithm (-PINNs)}   & Stiff    & Helmholtz   & LQG   & NS  \\ 
\midrule
KP ($\nu=\frac{1}{2}$)      & \textbf{0.2971$s$}             & \textbf{5.3152$s$}                & \textbf{90.7398$s$}     & \textbf{26.9730$s$}      \\
RKHS ($\nu=\frac{1}{2}$)    & 2.6984$s$             & 42.1662$s$               & 431.0374$s$       & 56.9642$s$  \\
\bottomrule
\end{tabular}
}
\caption{Time of computing $\mathbf{K}^{-1}$ ($\nu=\frac{1}{2}$ as example). $N_{\mathcal{B}}+N_{\mathcal{L}}=2000$ for Stiff equation. $N_{\mathcal{B}}+N_{\mathcal{L}}=200\times 300$ for Helmholtz equation. $N_{\mathcal{B}}+N_{\mathcal{L}}=50\times 50\times 40$ for LQG equation. $N_{\mathcal{B}}+N_{\mathcal{L}}=60\times 40\times 20$ for NS equation.}
\label{k-time}
\end{table}

\section{Conclusion} \label{s:conclusion}
This work proposes KP-PINNs, an efficient method for solving differential equations, with stability guarantees for second-order linear elliptic and HJB equations under certain conditions. Numerical experiments demonstrate that KP-PINNs outperform $L_2$-PINNs and Sobolev-PINNs for both forward and inverse problems, significantly accelerating the computation of the inverse of kernel matrix $\mathbf{K}$ through the KP technique. Future directions include extending KP-PINNs to handle sparse or irregular data, optimizing neural network architecture, and refining kernel selection in the loss function to improve accuracy and robustness. These advancements aim to enhance KP-PINNs' versatility in solving complex differential equations.

\section*{Acknowledgments}
This research is supported by Guangdong Provincial Fund - Special Innovation Project (2024KTSCX038).

\bibliographystyle{named}
\bibliography{ijcai25}


\end{document}